%% file: main.tex
\documentclass[a4paper]{lipics-v2021}
\pdfoutput=1 
\usepackage{graphicx,xcolor,amsmath,amsthm,amsfonts} 
\usepackage{xspace}
\usepackage{todonotes}
\usepackage[ruled,noend]{algorithm2e}
\usepackage{tikz}
\usetikzlibrary{automata, positioning,bbox, calc, arrows.meta, shapes}

\hideLIPIcs  %

\title{Precoloring extension with demands on paths}
\author{Arun Kumar Das}{School of Computer and Information Sciences, University of Hyderabad}{arunkumardas@uohyd.ac.in}{https://orcid.org/0000-0002-3645-4210}{}
\author{Michal Opler}{Faculty of Information Technology, Czech Technical University in Prague}{michal.opler@fit.cvut.cz}{https://orcid.org/0000-0002-4389-5807}{This work was supported by the Czech Science Foundation Grant no. 24-12046S.}
\author{Tomáš Valla}{Faculty of Information Technology, Czech Technical University in Prague}{tomas.valla@fit.cvut.cz}{https://orcid.org/0000-0003-1228-7160}{This work was supported by the Czech Science Foundation Grant no. 24-12046S.}

\Copyright{Arun Kumar Das, Michal Opler, Tomáš Valla} 
\authorrunning{A. K. Das, M. Opler, T. Valla}
\funding{This work was co-funded by the European Union under the project Robotics and advanced industrial production (reg. no. CZ.02.01.01/00/22\_008/0004590).}

\nolinenumbers 

\EventEditors{Ho-Lin Chen, Wing-Kai Hon, and Meng-Tsung Tsai}
\EventNoEds{3}
\EventLongTitle{36th International Symposium on Algorithms and Computation (ISAAC 2025)}
\EventShortTitle{ISAAC 2025}
\EventAcronym{ISAAC}
\EventYear{2025}
\EventDate{December 7--10, 2025}
\EventLocation{Tainan, Taiwan}
\EventLogo{}
\SeriesVolume{359}
\ArticleNo{43}

\newcommand{\defproblem}[3]{
	\vspace{3mm}
	\noindent\fbox{
		\begin{minipage}{.95\columnwidth}
			#1\newline
			\textbf{Input:} #2\\
			\textbf{Question:} #3
		\end{minipage}
	}
	\vspace{3mm}
}

\newcommand{\DPED}{\textsc{DPED}\xspace}

\newcommand{\LCD}{\textsc{LCD}\xspace}
\newcommand{\MSS}{\textsc{MSS}\xspace}
\newcommand{\DPE}{\textsc{DPE}\xspace}
\newcommand{\DLC}{\textsc{DLC}\xspace}

\newcommand{\NP}{{\sf NP}\xspace}
\newcommand{\Wone}{{\sf W[1]}\xspace}

\newcommand{\true}{\mathtt{true}}
\newcommand{\false}{\mathtt{false}}

\DeclareMathOperator{\dist}{dist}
\DeclareMathOperator{\pos}{pos}



\begin{document}

\maketitle

\begin{abstract}
Let $G$ be a graph with a set of precolored vertices, and let us be given an integer distance parameter $d$
and a set of integer \emph{demands} $d_1,\dots,d_c$.
The \emph{Distance Precoloring Extension with Demands} (\DPED) problem
is to compute a vertex $c$-coloring of $G$ such that the following three conditions hold:
(i) the resulting coloring respects the colors of the precolored vertices,
(ii) the distance of two vertices of the same color is at least $d$,
and (iii) the number of vertices colored by color $i$ is exactly $d_i$.
This problem is motivated by a program scheduling in commercial broadcast channels
with constraints on content repetition and placement, which leads precisely to the \DPED problem for paths.

In this paper, we study \DPED on paths and present a polynomial time exact
algorithm when precolored vertices are restricted to the two ends of the path
and devise an approximation algorithm for \DPED with an additive approximation
factor polynomially bounded by $d$ and the number of precolored vertices.
Then, we prove that the \emph{Distance Precoloring Extension} problem on paths, a less
restrictive version of \DPED without the demand constraints, and then \DPED itself, is NP-complete.
Motivated by this result, we further study the parameterized complexity of \DPED on paths.
We establish that the \DPED problem on paths is $W[1]$-hard when parameterized by the number of colors and the distance.
On the positive side, we devise a fixed parameter tractable (FPT)
algorithm for \DPED on paths when the number of colors, the distance, and the
number of precolored vertices are considered as the parameters.
Moreover, we prove that \emph{Distance Precoloring Extension} is FPT parameterized by the distance.
As a byproduct, we also obtain several results for the \emph{Distance List Coloring} problem on paths.

\end{abstract}

\section{Introduction}
\input{introduction.tex}

\section{Polynomial algorithm for paths with precolored end-vertices}
\label{sec:endprecolor}
\input{greedy}

\section{Few colors and small distance}
\label{sec:wonehardness}
\input{colors-distance}


\section{Few precolored vertices}
\label{sec:fewprecolored}
\input{few-precolored}

\section{Coloring without demands}
\label{sec:nodemands}
\input{no-demands}
\section{Acknowledgements}
We would like to thank
Václav Blažej, Dušan Knop, Josef Malík and Ondřej Suchý,
with whom we held preliminary discussions and obtained several observations relevant to this research.

\bibliographystyle{plainurl}
\bibliography{ref}

\appendix
\input{appendix}

\end{document}

%% file: introduction.tex
\emph{Vertex coloring} of graphs is one of the most studied and fundamental problems in structural and algorithmic graph theory. 
A \emph{k-coloring} of a given graph partitions its vertices into $k$ subsets (color classes) such that there is no edge between two members in the same partition.
Graph coloring problems~\cite{jensen2011graph} have a rich research
history~\cite{demange2015some,kratochvil1999new,woodall2001list,broersma2010coloring},
comprising numerous variants and results, due to their significant applications in
scheduling~\cite{bhattacharya2018dynamic}, resource allocation~\cite{osawa2018algorithms}, compiler design~\cite{smith2004generalized}, computational biology~\cite{Khor10},
network analysis~\cite{barenboim2009distributed}, and geography~\cite{freimer2000generalized}.
We study a particular variant of the vertex coloring problem with three additional constraints.
Some of the vertices are already given precolored on input, and the sought coloring must respect these colors.
Additionally, we are given required sizes of all color classes (called demands), and finally, we require any two vertices of the same color to lie at a distance of at least~$d$ from each other.
A vertex coloring that satisfies the last condition is called a \emph{$d$-distance coloring}.
Hence, classical vertex colorings are exactly $1$-distance colorings.
Alternatively, $d$-distance colorings of a graph~$G$ correspond to classical vertex colorings of the $d$th power of~$G$.

Formally, the problem we consider is stated as follows.
	
\defproblem{\textsc{Distance Precoloring Extension with Demands (\DPED)}}
{A graph $G = (V,E)$ on $n$ vertices, a set of colors $C$, a non-negative integer~$d$, a partial pre-coloring $\gamma'\colon A \to C$ for some $A \subseteq V$, and a demand function $\eta \colon C \to [n] $.}
{Is there a $d$-distance coloring $\gamma \colon V \to C$ that
\begin{itemize}
	\item $\gamma$ extends $\gamma'$ (i.e., $\gamma(v) = \gamma'(v)$ for every $v \in A$), and
	\item for every $c \in C$, the number of vertices newly colored by color $c$ is exactly $\eta(c)$, i.e., $|\{v \in V \setminus A \mid \gamma(v)=c\}| = \eta(c)$.  
\end{itemize}
}


We address the problem for paths and disjoint unions of paths.
Our research is motivated by a real-world scheduling problem of a television broadcasting company:
A commercial block of a daily broadcast consists of $n$ slots ($n$ vertices of a path),
where some slots have already been allocated the type of commercial (precolored vertices).
The commercials are already paid for (the demands are given), and no two commercials of the same type may be allocated close together.
The task is thus to find an assignment of commercials to slots (color the
vertices) such that no two commercials of the same type are too close (vertices being too close on the path must be colored by distinct colors), and all of the commercials (demands) are used.
This original setting leads to a lot of interesting generalizations.

A notable amount of previous research has been done in the direction of the coloring of
graphs and most of the problems are proved to be \NP-hard depending on the type of the input.
There are numerous complexity results
for different variants of the coloring problems, for example, in terms of parameters like treewidth~\cite{fellows2011complexity},
distance to cluster and co-cluster~\cite{GomesGS23},
number of colors and maximum degree~\cite{de2019parameterized},
clique modulator~\cite{gutin2021parameterized},
or vertex cover~\cite{Jaffke2023}.
Bir{\'o} et al.~\cite{biro1992precoloring} introduced the precoloring extension problem for interval graphs.
They proved that the problem is polynomial time solvable when each color is used at most once for precoloring, and \NP-complete when they are used twice or more.
They also extended the results for graphs with bounded treewidth.
They mentioned the problem as a variant of the \textsc{List Coloring} problem~\cite{erdos1979choosability}.
In this problem, each vertex is assigned with a list of colors, and they must be colored with a color from the assigned list.
The graph is called \emph{L-list colorable} if there exists a valid vertex
coloring which chooses every color from the assigned lists~$L$, and it is called \emph{k-choosable} if this is possible for any assignment of lists of size at least~$k$.
List coloring has been studied on trees~\cite{bodlaender2022list},
planar graphs~\cite{voigt1993list}, and many other graph classes~\cite{Skrekovski98}.

Yet another studied variant of coloring, the \textsc{Equitable Coloring} problem~\cite{meyer1973equitable},
requires the difference between the total number of vertices colored with two different colors to be at most one.
Hence \textsc{Coloring with Demands} can be seen as a direct generalization of this notion.
In practice, it is common to combine multiple constraints for coloring problems in order to address various applications, e.g., Pelsmajer~\cite{pelsmajer2009equitable} studied the problem of equitable list coloring.
In the \textsc{Distance Coloring} problem, monochromatic vertices cannot occur within a given distance from each other.
This generalization of coloring has also been studied extensively~\cite{sharp2007distance,agnarsson2003coloring,agnarsson2000powers}.
As stated earlier, the problem of coloring paths is motivated by scheduling problems
and it has been studied for many other variants, including call scheduling~\cite{erlebach2001complexity},
nonrepetitive list colourings~\cite{grytczuk2011nonrepetitive}, \emph{radio $k$-colorings}~\cite{chartrand2004radio} etc. 

\begin{figure}[h]
    \centering
    \includegraphics[width=\linewidth]{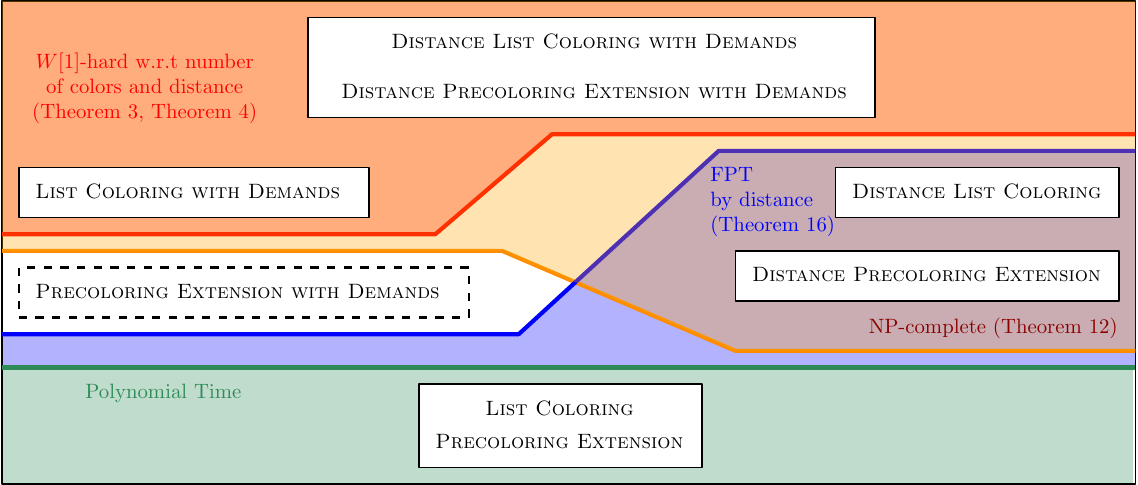}
    \caption{A pictorial representation of the results.}
    \label{fig:res}
\end{figure}

In this paper, we study the complexity of \DPED when $G$ is the union of disjoint paths.
We prove that the decision version of the problem is \NP-complete when the number of colors, the number of precolored vertices, and the minimum distance $d$ between two vertices of the same color are part of the input.
This hardness motivates the study of the problem in terms of parameterized complexity.
We reduce the problem to the problem of \textsc{List Coloring with Demands}.
We show that both problems are \Wone-hard when parameterized by distance and the number of colors.
We prove that the hardness holds for the list coloring problem even if we drop the condition of distance coloring.
However, the problem admits an FPT algorithm when we drop the demands and consider only the distance,
implying that \textsc{Distance Precoloring Extension} on paths is also FPT by distance.
Both problems remain \NP-complete when the number of colors, the number of precolored vertices, and the distance are considered a part of the input.
On the positive side, we prove that \DPED is polynomial time solvable on a single path when the precolored vertices appear only at the ends.
Further, we show that \DPED on path instances is FPT when the number of colors,
the distance parameter, and the number of precolored vertices are considered as parameters simultaneously.
Then we present an approximation algorithm that runs in time $pd^{O(d)} + O(n\log c)$
and provides an additive approximate solution with an error of $O(d^2 \cdot p)$.
Finally, we establish the \NP-completeness of the problems of \textsc{Distance List Coloring}
and \textsc{Distance Precoloring extension} on paths.
The relationships between the problems and the results are illustrated in Figure~\ref{fig:res}.
We pose the complexity of the \textsc{Precoloring Extension with Demands} on paths as an interesting open problem.


The paper is organized as follows.
Section~\ref{sec:endprecolor} presents a polynomial time solution for the \DPED problem on a path having only its ends precolored.
In Section~\ref{sec:wonehardness} we establish the \Wone-hardness of \DPED when the parameters are the number of colors and the distance.
Section~\ref{sec:fewprecolored} contains results for special cases, when the problem parameters (number of
precolored vertices, number of colors, distance) are in some sense limited: we present an approximation algorithm with small additive error and an exact FPT algorithm.
Finally, Section~\ref{sec:nodemands} presents dynamic programming algorithm and \NP-completeness results
for the variant without the demands.

%% file: greedy.tex
As our first result, we show that \textsc{Distance Precoloring Extension with Demands} can be solved in polynomial time whenever we allow precolored vertices only at the ends of the path.
In this setting, we are coloring a path $G$ with vertices $v_1, \dots, v_n$ with a precoloring~$\gamma'$ that assigns colors to some initial segment $v_1, \dots, v_s$ and some final segment $v_t, \dots, v_n$.
We say that such an instance of~\DPED is \emph{end-precolored}.
We describe a greedy algorithm that solves end-precolored instances of~\DPED in polynomial time, regardless of both the number of colors and the distance parameter.

The algorithm colors the vertices in their left-to-right order along the path, starting with the first uncolored vertex.
For each vertex, it computes a set of feasible colors for the current vertex and chooses the one with the highest remaining demand. In case of a tie, the algorithm
chooses the color that appears earlier in the precoloring of the final segment of the path.
We remark that to compute the set of feasible colors for vertex~$v$, it is necessary to exclude not only colors of the previous $d$ vertices to the left of~$v$ but also the precolored endsegment in case $v$ is at distance at most $d$ from it. 
The technique is formally
described in Algorithm~\ref{algo:appx-end-seg}.

\begin{algorithm}
	\DontPrintSemicolon
	\newcommand\mycommfont[1]{\textit{#1}}
	\SetCommentSty{mycommfont}
	
	\SetKwComment{Comment}{$\triangleright$\ }{}
	\SetKwComment{Commentl}{$\triangleright$\ }{$\triangleleft$}  
	\caption{Greedy algorithm for end-precolored \DPED \label{algo:appx-end-seg}}
	\KwIn{End-precolored path instance $(G, C, d, \gamma', \eta)$ of \DPED, where $G$ has vertices $v_1, \dots, v_n$ and $\gamma'$ precolors the segments $v_1, \dots, v_s$ and $v_t, \dots, v_n$.}
	\KwOut{A $d$-distance coloring $\gamma$ of $G$ that extends $\gamma'$ and satisfies the demands $\eta$, if it exists,  and ``\texttt{null}'' otherwise.}
	\For(\Comment*[f]{compute tie-breakers depending on the final segment}){$a \in C$}{$\pos(a) \gets$ minimum $j \ge t$ such that $\gamma'(v_j) = a$, or $\infty$ if no such $j$ exists.}
	$\gamma \gets \gamma'$ \Comment*{initialize $\gamma$ with the precoloring~$\gamma'$}
	\For(\Comment*[f]{iterate over non-precolored vertices}){$i\gets s+1$ \KwTo $t-1$}{
		$C_f \gets C \setminus \gamma(\{v_{\max(i-d,1)}, \dots, v_{\min(i+d, n)}\})$ \Comment*{the set of feasible colors for $v_i$}
		\If{$C_f=\emptyset$ \textbf{\normalshape or} $\eta(a) = 0$ for every $a \in C_f$}{\KwRet{\texttt{null}} \Comment*{no feasible color with positive demand}}
		$a \gets $ a color in~$C_f$ with the largest demand~$\eta(a)$ that, moreover, minimizes $\pos(a)$\\
		$\gamma(v_i) \gets a$\Comment*{$v_i$ receives the color~$a$}
		$\eta(a) \gets \eta(a) - 1$ \Comment*{demand of $a$ is decreased}
	}
	\KwRet{$\gamma$}
\label{apxalgo}
\end{algorithm}

\begin{theorem}
	\label{thm:greedy-algo}
	Algorithm~\ref{apxalgo} solves \DPED on end-precolored paths in time $O(n \log c)$ where $n$ is the number of vertices and $c$ is the number of colors.
\end{theorem}

\begin{proof}
For contradiction, assume that Algorithm~\ref{apxalgo} fails, that is, there exists a feasible solution to the instance, and the algorithm does not find any.
Let us denote the vertices of the path~$P$ as $v_1,v_2,\dots,v_n$ from left to the right
and let $R=\{v_t,\dots,v_n\}$ be the right precolored segment of $P$.
Let the algorithm produce a partial coloring $\gamma$ and let $\gamma'$ be a feasible solution where
the index $i$ satisfying $\gamma(v_j)=\gamma'(v_j)$ for all $j<i$ and $\gamma(v_i)\ne\gamma'(v_i)$
is the maximum possible.
We are going to modify $\gamma'$ such that we obtain another feasible solution that agrees with $\gamma$ even on $\gamma(v_i)$.

Let $a=\gamma(i)$ and $b=\gamma'(i)$.
By \emph{alternating $ba$-sequence starting at $v_i$} we denote a maximal subsequence $S=(s_1,s_2,\dots s_k)$
of the vertices $v_i,v_{i+1},\dots,v_n$ such that $s_1=v_i$, $\gamma'(s_j)=b$ for odd $j$ and $\gamma'(s_j)=a$ for even $j$,
the distance of $s_j$ and $s_{j+1}$ is at most $d$ for every $j \in [k-1]$,
and no vertex $s_j$ is precolored.
Observe that when the algorithm was assigning color to $v_i$, the remaining demand of the color $a$
was at least the remaining demand of the color $b$, and that there is no color
$a$ in the distance less than $d$ of the vertex $v_i$.
Also note that for an alternating $ba$-sequence $s_1,\dots,s_k$, there can be no vertex of $P$
between any $s_i$ and $s_{i+1}$ colored by $a$ or $b$, as this would break the distance property.
Let us distinguish several cases.

First, consider the case when the last element of the sequence $S$ is closer than $d$
to the right precolored part $R$ of $P$.
In this case the sequence $S$ must contain all occurrences of color $a$ and $b$
among the vertices from $v_i$ to $v_n$ in $\gamma'$.
This implies that $S$ is an even sequence, as the demand of $a$ was at least the demand of $b$.
If the vertices from $R$ contain neither color $a$ nor color $b$,
we can freely swap the colors $a$ and $b$ in $S$.
Suppose now the vertices from $R$ contain color $a$ (or $b$).
If there is no occurrence of $a$ or $b$ in $R$ that is within distance $d$ of $s_k$,
then the colors $a$ and $b$ can be swapped in $\gamma'$.
If the first occurrence of $a$ or $b$ in $R$ is within distance $d$ of $s_k$,
it must be $b$ because $s_k$ is colored $a$ by $\gamma'$,
but this contradicts the fact that the algorithm has colored $s_i$ with $a$ instead of $b$.

Next, assume that the alternating $ba$-sequence $S$ starting at $v_i$ has even length and its last element
is more than $d$ vertices away from the precolored right part~$R$ of $P$.
We can freely swap the colors $a$ and $b$ at the vertices of $S$ as there is no conflict both to the left and right of $S$, thus obtaining a feasible solution $\gamma''$ with greater index $i$.

It remains to solve the case when $S$ has odd length and its last element is more than $d$ vertices far
from the precolored right part~$R$ of $P$.
The sequence $S$ contains one more color $b$ than $a$ and as the algorithm has chosen the color $a$ for $v_i$
(which means the remaining demand of $b$ was not greater than $a$),
there must exists alternating $ab$-sequence $S'$ starting at some vertex $w$ of odd length, $S'$ disjoint from $S$,
which thus contains one more color $a$ than $b$.
We may assume that $S'$ is not within distance $d$ to the right precolored part $R$:
if this is the case, we either obtain a contradiction if $a$ occurs before $b$ in $R$ and $a$ and $b$ had the same demand,
or we can recolor $S$ and $S'$ without problems. 
We may now swap the colors $a$ and $b$ both in $S$ and $S'$,
which does not change the total number of colors used. This produces a feasible solution $\gamma''$ with greater index $i$.

We obtained the contradiction in all cases, which shows the correctness of the algorithms.
The time complexity $O(n\log c)$ follows from using a smarter data structure for assigning the available colors: we may dynamically maintain a binary heap containing the set of feasible colors~$C_f$ with the key being a combination of the remaining demand
with position ($\pos(\cdot)$) and the value stored being the color label.
\end{proof}

%% file: colors-distance.tex
In this section, we focus on solving general instances of \DPED when both the number of colors and the distance parameter are not too large.
We stress that this additional assumption is very natural, e.g., for the original scheduling motivation.

First, we observe that we can solve \DPED in this regime in polynomial time,
provided that the number of colors and the distance parameter $d$ are constant.
The algorithm follows a standard dynamic programming approach for coloring and thus, we provide here only a brief sketch and include all details in Appendix~\ref{apx:LCD-W1h}.
\begin{proposition}
\label{pro:dp}
An instance $(G, C, d, \gamma', \eta)$ of \DPED on paths can be solved in time $c^{O(d)} \cdot n^{c+1}$ where $n$ is the number of vertices and $c$ is the number of colors.
\end{proposition}
\begin{proof}[Proof sketch]
Let $\gamma_i$ be a $d$-distance coloring of an initial segment $\{v_1, \dots, v_i\}$ of the path~$G$.
The signature of~$\gamma_i$ consists of the colors of the rightmost $d$~vertices together with the frequencies of individual colors as used by~$\gamma_i$.
The algorithm sequentially computes for each $i \in [n]$ all possible signatures of $d$-distance colorings of $\{v_1, \dots, v_i\}$ that, moreover, extend the precoloring~$\gamma'$.
At the end, this suffices to check whether there is a $d$-distance coloring of $\{v_1, \dots, v_n\}$ that extends $\gamma'$ and the frequencies of colors used equal the demands.
\end{proof}

It is only natural to ask whether we can achieve a polynomial runtime in this regime where the degree of the polynomial is independent of both $c$ and $d$.
We answer this question negatively, under standard theoretical assumptions, by showing that \DPED is \Wone-hard with respect to both of these parameters.

\begin{theorem}
\label{thm:DPED-Wh}
\DPED is \Wone-hard on paths with respect to the number of colors and the distance~$d$.
\end{theorem}


We first show \Wone-hardness of a related coloring problem, namely \textsc{List Coloring with Demands}. 
We define the problem formally as follows.

\defproblem{\textsc{List Coloring with Demands (LCD)}}
{A graph $G = (V,E)$ on $v$ vertices, a set of colors $C$, a function $L\colon V \to 2^C$ that assigns a list of colors~$L(v)$ to each vertex, and a demand function $\eta \colon C \to [n]$.}
{Is there an $L$-list coloring $\gamma \colon V \to C$ that satisfies the demands~$\eta$?}

An instance $(G, C, \eta, L)$ of \LCD is a \emph{path instance} if the graph $G$ is a disjoint union of paths.
Moreover, a path instance of \LCD is \emph{non-alternating} if there does not exist a color $c \in C$ and three consecutive vertices $v_{i-1}, v_{i}, v_{i+1}$ on some path such that $c \notin L(v_{i})$ and simultaneously $c \in L(v_{i-1}) \cap L(v_{i+1})$.

The \LCD problem has already been shown to be \Wone-hard on paths with respect to the number of colors by Gomes, Guedes, and dos Santos~\cite{GomesGS23} albeit under a different name of \textsc{Number List Coloring}.
However, the produced instances of \LCD therein are very far from having the non-alternating property that is crucial in proving Theorem~\ref{thm:DPED-Wh}.
We devise a novel reduction from a multidimensional version of the subset problem that, moreover, produces non-alternating path instances of \LCD.
Due to the space constraints, we postpone the proof to Appendix~\ref{apx:LCD-W1h}.


%


\begin{theorem}
\label{thm:LCD-W1h}
\LCD is \Wone-hard with respect to the number of colors even when restricted to non-alternating path instances.
\end{theorem}

Theorem~\ref{thm:DPED-Wh} is obtained through a polynomial-time reduction from \LCD on non-alternating path instances.
Crucially, both the distance parameter and the number of colors in the produced \DPED instance are linear in the number of colors in the original \LCD instance.

\begin{proof}[Proof of Theorem~\ref{thm:DPED-Wh}]
For technical reasons, we first modify the \LCD instance to guarantee that it is a single path, and both of its endpoints have a list identical to its only neighbor~$v$.
That is achieved by introducing two extra colors $a, b$ and concatenating all paths in arbitrary order with two extra vertices in between consecutive pairs and two extra vertices at the very beginning and end where we add $\{a,b\}$ to the list of all vertices and set $L(v) =\{a,b\}$ for each new vertex.
Finally, we set $\eta(a) = \eta(b) = p+1$ where $p$ is the total number of paths in~$G$.
It is easy to see that the modified instance is equivalent.
Furthermore, the modified instance remains non-alternating because there are two extra new vertices between any pair of original paths, and the colors $a, b$ are included in the list of every vertex.

From now on, we assume that $G$ is a single path with vertices $v_1, \dots, v_n$ and edges $e_1, \dots, e_{n-1}$ such that $e_i = \{v_i, v_{i+1}\}$ for each $i \in [n -1]$.
Moreover, we have $L(v_1) = L(v_2)$ and $L(v_{n-1}) = L(v_n)$.

Now, we show that the lists of such a non-alternating \LCD path instance can be alternatively represented by lists of forbidden colors for each edge.
It is crucial that the instance of \LCD is non-alternating, as otherwise, this may not be possible.
Moreover, we will also guarantee that any fixed color can be contained in the forbidden sets of at most two consecutive edges.

\begin{claim}
	\label{claim:colors-on-edges}
	We can compute in polynomial time an assignment of color sets to edges $F\colon E(G) \to 2^C$ such that for every vertex $v$, we have $L(v) = C \setminus (F(e) \cup F(e'))$ where $e, e'$ are the two edges incident to~$v$, or $L(v) = C \setminus F(e)$ if $v$ is incident to a single edge~$e$.
	Moreover, there are no three consecutive edges $e_{i-1}$, $e_i$ and $e_{i+1}$ such that the intersection $F(e_{i-1}) \cap F(e_i) \cap F(e_{i+1})$ is non-empty.
\end{claim}
\begin{claimproof}
	We define the assignment inductively along the path.
	We set $F(e_1) = C \setminus L(v_1)$ and $F(e_{n-1}) = C \setminus L(v_n)$.
	Recall that we first modified the instance to guarantee that each endpoint of the path has an identical list to its neighbor, i.e., $L(v_1) = L(v_2)$ and $L(v_{n-1}) = L(v_n)$.
	For $i \in \{2, \dots, n - 2\}$,  we let $F(e_i)$ contain a color $c \in C$ if and only if $c \notin L(v_i) \cup L(v_{i+1})$ and moreover, either (i) $c \notin F(e_{i-1})$, or (ii) $c \in L(v_{i+2})$.
	See Figure~\ref{fig:edge-forbid}.
	Clearly, $F$ can be computed in polynomial time.
	\begin{figure}
		\centering
		\includegraphics[width=\linewidth]{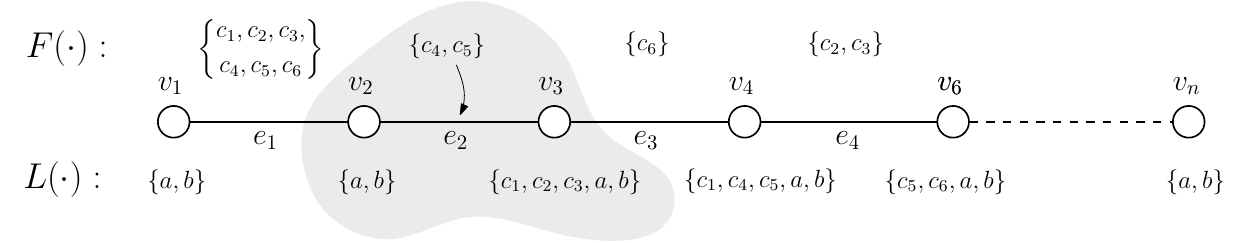}
		\caption{Representation of LCD with forbidden colors  when $C = \{c_1,c_2, \dots, c_6,a,b\}$.}
		\label{fig:edge-forbid}
	\end{figure}
	
	Let us verify the properties of~$F$.
	For the endpoints $v_1$ and $v_n$, we have $L(v_1) = C \setminus F(e_1)$ and $L(v_n) = C \setminus F(e_{n-1})$ by definition. 
	Now, fix arbitrary $i \in \{2, \dots, n - 1\}$ and a color $c \in C$.
	If we have $c \in L(v_i)$, then neither $F(e_{i-1})$ nor $F(e_i)$ contain $c$ by definition and $c \in C \setminus (F(e_{i-1}) \cup F(e_i))$.
	On the other hand, if $c \notin L(v_i)$, it follows from the non-alternating property that either $c \notin L(v_{i-1})$ or $c \notin L(v_{i+1})$.
	We consider the three possible cases.
	First, assume that $c$ appears in none of the lists $L(v_{i-1})$, $L(v_i)$, and $L(v_{i+1})$.
	Then either $c \in F(e_{i-1})$ or we have $c \notin F(e_{i-1})$ which triggers the condition (i) in the definition of~$F(e_i)$ and either way, we get that $c \in F(e_{i-1}) \cup F(e_i)$.
	Second, assume that $c$ appears in the list $L(v_{i+1})$ but not in $L(v_{i-1})$.
	In this case, we have $c \in F(e_{i-1})$ by condition (ii) in the definition of $F(e_{i-1})$ and again, we see that $c \in F(e_{i-1}) \cup F(e_i)$.
	Finally, assume that $c$ appears in the list $L(v_{i-1})$ but not in $L(v_{i+1})$.
	In this case, we have $c \notin F(e_{i-1})$, which implies $c \in F(e_i)$ by condition~(i) in the definition.
	In all cases, we obtained  $c \in F(e_{i-1}) \cup F(e_i)$.
	
	Assume for a contradiction that there are three consecutive edges $e_{i-1}$, $e_i$ and $e_{i+1}$ such that the intersection $F(e_{i-1}) \cap F(e_i) \cap F(e_{i+1})$ is non-empty and let $c$ be an arbitrary color in the intersection.
	However, observe that in the definition of $F(e_i)$ neither condition (i) nor (ii) holds for~$c$ and thus, we must have $c \notin F(e_i)$.
\end{claimproof}

With this assignment~$F$, we are finally ready to define the path instance $(G', C', d, \gamma', \eta')$ of \DPED.
Let $t$ denote the size of the color set~$C$ and let $C = \{c_1, \dots, c_t\}$ be an arbitrary enumeration of the colors.
We set $d = 2t + 1$.
We take the graph $G'$ to be a path of length $nd + 1$ on vertices $v'_1, \dots, v'_{nd  + 1}$ in this order along the path.
To simplify the exposition, we partition the vertices into two distinct sets.
Every vertex $v'_{(k-1) d + 1}$ for some $k \in [n]$ is a \emph{main vertex} and we denote it alternatively by~$x_k$.
All the remaining vertices are \emph{auxiliary}.
For $i \in [n-1]$, $j \in [t]$ and $\alpha \in [2]$, the vertex $v'_{(i-1) d + 2j + \alpha - 1}$ is an auxiliary vertex \emph{associated to color~$c_j$} and we alternatively denote it by $y^i_{j, \alpha}$.
Observe that every two consecutive main vertices $x_i$ and $x_{i+1}$ are separated by exactly $2t$ auxiliary vertices that are grouped into pairs associated to colors $c_1, \dots, c_t$ in this order.
More specifically, the segment between $x_i$ and $x_{i+1}$ contains the auxiliary vertices $y^i_{j,\alpha}$ for all $j \in [t]$ and $\alpha \in [2]$, ordered lexicographically by the pairs $(j, \alpha)$.
See Figure~\ref{fig:lcdtodped}.

\begin{figure}
	\centering
	\includegraphics[width=\linewidth]{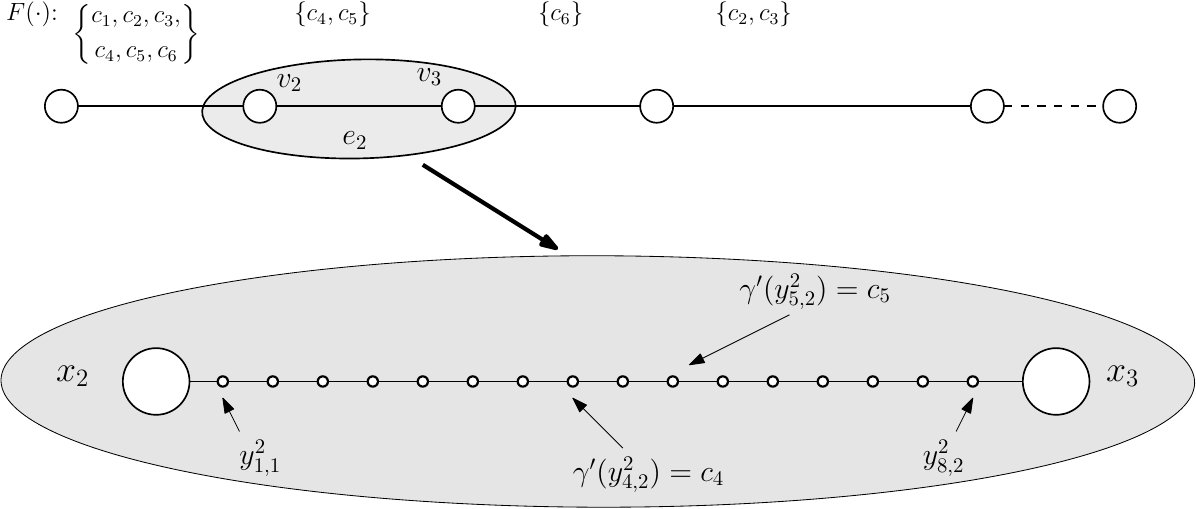}
	\caption{Construction of $G'$ from $G$. Extending from Figure $2$, $\gamma'(y^2_{j,\alpha})=\gamma^\star(y^2_{j,\alpha})$, except for $y^2_{4,2}$ and $y^2_{5,2}$.}
	\label{fig:lcdtodped}
\end{figure}

Our goal is to define a pre-coloring of all the auxiliary vertices such that there is a one-to-one correspondence between feasible colorings of the main vertices and feasible colorings of the original \LCD instance.
As a final ingredient, we need to be able to assign extra colors to certain vertices without affecting the rest of the instance.
To that end, we set $C' = C \cup C^\star$ where $C^\star = \{c^\star_0, \dots, c^\star_d\}$ is a set of \emph{auxiliary colors} disjoint from~$C$.
Moreover, we define an auxiliary coloring $\gamma^\star\colon V(G')\to C^\star$ where $\gamma^\star(v'_i) = c^\star_{i \bmod (d+1)}$ for each $i \in [nd + 1]$.
Informally, the coloring~$\gamma^\star$ simply cyclically iterates through the sequence of auxiliary colors $c^\star_0, \dots, c^\star_d$ along the path, starting with $c^\star_1$.
Clearly, $\gamma^\star$ is a proper $d$-distance coloring of~$G'$.

We define the pre-coloring $\gamma'$ of every auxiliary vertex.
For $i \in [n-1]$ and $j \in [t]$, we set
\begin{align*}
	\gamma'(y^i_{j,1}) &= \begin{cases}
		c_j &\text{if $c_j \in F(e_i)$ and  $c_j \notin F(e_{i-1})$.}\\
		\gamma^\star(y^i_{j,1}) &\text{otherwise.}
	\end{cases}\\
	\gamma'(y^i_{j,2}) &= \begin{cases}
		c_j &\text{if $c_j \in F(e_i) \cap F(e_{i-1})$.}\\
		\gamma^\star(y^i_{j,2}) &\text{otherwise.}
	\end{cases}
\end{align*}
where we additionally define $F(e_0) = \emptyset$ for simplicity.

First, observe that the pre-coloring $\gamma'$ exactly encodes the assignment of forbidden colors to edges.
\begin{observation}
	\label{obs:gamma-F}
	A color $c_j \in C$ appears in the pre-coloring $\gamma'$ on the segment between main vertices $x_i$ and $x_{i+1}$ if and only if $c_j \in F(e_i)$.
\end{observation}

Moreover, we show that $\gamma'$ distributes the occurrences of a fixed color at a sufficient distance from each other.

\begin{claim}
	\label{claim:gammap-valid}
	The pre-coloring $\gamma'$ does not contain any monochromatic pair of vertices at a distance at most $d$.
\end{claim}
\begin{claimproof}
	Assume for a contradiction that $\gamma'$ does contain such a pair of monochromatic vertices.
	We have already observed that $\gamma^\star$ is a proper $d$-coloring on a disjoint set of auxiliary colors and thus, this pair must use some color $c_j \in C$.
	We cannot have $\gamma'(y^i_{j,1}) = \gamma'(y^i_{j,2}) = c_j$ for any fixed $i,j$ since the respective conditions in the definition of $\gamma'$ are mutually exclusive.
	
	For any different $i, i' \in [n-1]$ with $i < i'$ and arbitrary $\alpha, \alpha' \in [2]$, the distance between $y^i_{j,\alpha}$ and $y^{i'}_{j, \alpha'}$ is at most~$d$ if and only if $i' - i = 1$ and $\alpha' \le \alpha$.
	We cannot have $\gamma'(y^i_{j,1}) = \gamma'(y^{i+1}_{j,1}) = c_j$ by definition of $\gamma'(y^{i+1}_{j,1})$.
	So the only option left is that $\gamma'(y^i_{j,2}) = \gamma'(y^{i+1}_{j,\alpha}) = c_j$ for either $\alpha \in [2]$.
	But in that case, we have $c_j \in F(e_{i-1}) \cap F(e_i)$ since $\gamma'(y^i_{j,2})= c_j$ and simultaneously, $c_j \in F(e_{i+1})$ since $\gamma'(y^{i+1}_{j,\alpha})= c_j$.
	Thus, we reached a contradiction since $F$ does not contain the same color in the sets $F(e_{i-1}), F(e_i), F(e_{i+1})$ of three consecutive edges.
\end{claimproof}

To finish the construction, it remains to define the demand function $\eta'$.
We set $\eta'(c_j) = \eta(c_j)$ for every original color $c_j \in C$ and we set $\eta'(c^\star_j) = 0$ for every auxiliary color $c^\star_j \in C^\star$.

\subparagraph{Correctness (``only if'').}
Suppose that  $(G, C, \eta, L)$ is a yes-instance of \LCD and let $\rho\colon V(G) \to C$ be an $L$-list coloring that meets the demands.
We define a coloring $\gamma\colon V(G') \to  C'$ by simply setting $\gamma(x_i) = \rho(v_i)$ for every main vertex~$x_i$ and $\gamma(y^i_{j,\alpha}) = \gamma'(y^i_{j,\alpha})$ for every auxiliary vertex~$y^i_{j,\alpha}$.
The coloring $\gamma$ is an extension of $\gamma'$ by definition, and it meets the demands because the demand functions $\eta$ and $\eta'$ are equal when restricted to the color set~$C$.
Thus, it remains to show that $\gamma$ is a proper $d$-distance coloring.

Assume for a contradiction that $\gamma$ contains a pair of monochromatic vertices $u$ and $w$ at a distance at most~$d$.
At least one of these vertices has to be a main vertex since the colors of all auxiliary vertices are fixed by~$\gamma'$ and there is no such monochromatic pair of pre-colored vertices by Claim~\ref{claim:gammap-valid}.
First, assume that both $u$ and $w$ are main vertices.
Two main vertices are in distance at most~$d$ if and only if they are consecutive, i.e., we have without loss of generality $u = x_i$ and $w = x_{i+1}$ for some $i \in [n-1]$.
But then they cannot share the same color since $\gamma(x_i) = \rho(v_i)$, $\gamma(x_{i+1}) = \rho(v_{i+1})$ and $\rho$ is a proper list-coloring of~$G$.
Otherwise, suppose without loss of generality that $u$ is a main vertex~$x_i$ and $w$ is an auxiliary vertex.
The coloring~$\gamma$ can use only the colors from~$C$ on the main vertices and thus, we have $\gamma(x_i) = \gamma(w) = c_j$ for some $c_j \in C$.
Moreover, the vertex $w$ lies either on the segment between $x_{i-1}$ and $x_i$ or on the segment between $x_i$ and $x_{i+1}$ as those are the only auxiliary vertices at a distance at most~$d$ from $x_i$.
Observation~\ref{obs:gamma-F} then implies that we have either $c_j \in F(e_{i-1})$ or $c_j \in F(e_i)$.
But we know that $L(x_i) = C \setminus (F(e_{i-1}) \cap F(e_i))$ by Claim~\ref{claim:colors-on-edges}.
In particular, the color~$c_j$ does not appear in the list $L(v_i)$, which is a contradiction with $\rho$ being a proper list-coloring since $\rho(v_i) = \gamma(x_i) = c_j$.

\subparagraph{Correctness (``if'').}
Now suppose that $(G', C', d, \gamma', \eta')$  is a yes-instance of \DPED and let $\gamma\colon V(G') \to C'$ be a witnessing $d$-distance coloring, i.e., $\gamma$ extends the pre-coloring~$\gamma'$ and meets the demands given by~$\eta'$.
We define a coloring $\rho\colon V(G) \to C$ by simply restricting~$\gamma$ to the main vertices, that is we set $\rho(v_i) = \gamma(x_i)$.
First, observe that the range of~$\rho$ is indeed only the color set~$C$ since the demand $\eta'(c^\star_j)$ of any auxiliary color is zero and thus, $\gamma$ uses only the colors from~$C$ on the main vertices.
Moreover, we have $\rho(v_i) \neq \rho(v_{i+1})$ for every $i \in [n-1]$ because the main vertices $x_i$ and $x_{i+1}$ are at distance~$d$ in~$G'$.
The demands are also satisfied by~$\rho$ because the demand functions $\eta$ and $\eta'$ are equal when restricted to the color set~$C$.

It remains to check that $\rho$ is an $L$-list coloring, i.e., we have $\rho(v_i) \in L(v_i)$ for all $i \in [n]$.
Assume for a contradiction that this does not hold for a vertex~$v_i$ with $\rho(v_i) = \gamma(x_i) = c_j$.
By Claim~\ref{claim:colors-on-edges}, the color $c_j$ belongs to at least one of the sets $F(e_{i-1})$ and $F(e_i)$.
Observation~\ref{obs:gamma-F} implies that the color~$c_j$ is used by~$\gamma$ on some auxiliary vertex from the segment between $x_{i-1}$ and $x_{i+1}$.
But all these vertices are at a distance at most~$d$ from~$x_i$ and we reach a contradiction with $\gamma$ being a proper $d$-distance coloring.
\end{proof}

%% file: few-precolored.tex
In this section, we consider the regime when we are given only a few precolored vertices.
We present two effective algorithms for \DPED in this regime.
We first show that an approximation FPT algorithm with small additive error is possible even in the case of many colors.
Afterwards, we present an exact FPT algorithm when the number of colors is small as well.

\subsection{Approximation}
\label{subsec:approx}

Next, we show there is an approximation algorithm for \DPED with additive error,
which is bounded by a polynomial function of the distance parameter and the number of colors.
By having additive error $\alpha$ we mean that for the set of colors $C$ and demands $\eta\colon C\to[n]$
the algorithm produces coloring $\gamma\colon V\to C$ such that
$$
\sum_{c\in C} \left| \eta(c) - |\gamma^{-1}(c) \setminus A| \right| \le \alpha.
$$

\begin{theorem}
	Let us have an instance of \DPED on path $P$ of length $n$ with $p$ precolored vertices, and let $c=|C|\ge d+2$.
	There is an approximation algorithm that runs in time $pd^{O(d)} + O(n\log c)$ and outputs a coloring with an additive error at most $O(pd^2)$.
\end{theorem}

\begin{proof}
	The approximation algorithm is as follows.
	We ignore all precolored vertices and we run the greedy Algorithm~\ref{apxalgo} to color the path $P$.
	Let the precolored vertices be denoted $w_1,\dots,w_p$, respectively, along the path $P$, which
	consists of vertices $v_1,\dots,v_n$ from left to right.
	Let $b=2(d+1)^2$.
	Then we remove the colors in the neighborhood of diameter $b$ around each precolored vertex $w_i$
	(that is, for the vertices $v_j$ whose distance to some $w_i$ is at most $b$).
	The idea is now to reassign the colors to these discolored vertices while introducing only a small additive error.

	First we handle all pairs of vertices $(w_i,w_{i+1})$ such that $\dist(w_i, w_{i+1}) \le 2b$ as follows.
	If $|C|\ge 2d+1$, it is easy to
	find the coloring of the gap between $w_i$ and $w_{i+1}$ in a greedy manner.
	If $|C| < 2d+1$, we produce $d$-distance coloring of each segment of consecutive gaps of size at most~$2b$ by dynamic programming in time $pd^{O(d)}$.
	We may thus further assume that all remaining uncolored gaps between pairs of vertices $(w_i,w_{i+1})$
	have $\dist(w_i,w_{i+1})> 2b$.
	For a precolored vertex $w_i$, which equals to some $v_j$, let us denote by \emph{left neighborhood of $w_i$}
	those vertices $v_k$, where $\dist(v_k,v_j)\le b$ and $k<j$.
	Similarly, we define the right neighborhood as vertices $v_k$ where $\dist(v_k,v_j)\le b$ and $k>j$.
	We now describe the procedure that colors the diameter $b$ neighborhood of each $w_i$, first the left
	neighborhood and then the right neighborhood (provided that the respective neighborhood is still uncolored).
	Observe that if the right neighborhood of~$w_i$ and left neighborhood of~$w_{i+1}$ are both uncolored then they are disjoint as $\dist(w_i,w_{i+1})> 2b$.
	We split the $b$ vertices of the left neighborhood into $2(d+1)$ blocks $B_1,\dots,B_{2(d+1)}$ of size $d+1$.
	Let us denote the $d+1$ vertices immediately preceding block $B_1$ as $B_0$
	and let the $d+1$ vertices after $B_{2(d+1)}$ be denoted as $B_{2d+3}$.
	Here, the vertex $w_i$ is the first vertex of $B_{2d+3}$ and we choose the coloring of $B_{2d+3}$
	to consist of the precoloring of $w_i$ followed by arbitrarily chosen valid coloring of the rest of $B_{2d+3}$.
	Denote by $B_i^j$ the color of the vertex $u_j$ of the block $B_i=u_1,\dots,u_{d+1}$.

	We proceed by partial modification of the coloring of each block $B_i$ into the block $B_{i+1}$, for $i=0,\dots,2d+1$.
	Let the coloring of $B_0$ be, without loss of generality, $0,1,\dots,d$.
	For the block $B_1$, we copy the coloring from $B_0$, where we remove the
	occurrence of color $B_{2d+3}^1$, shift the colors to the right, and introduce a new first color to be $d+1$ (that is, an extra color).
	For the block $B_2$, we copy the coloring from $B_1$, where we replace the first color $d+1$ into $B_{2d+3}^1$.
	Note that by these two steps, we did not create any color conflict for vertices closer than $d$.
	For $i=2,\dots,d+1$ we now continue in a similar manner:
	let the coloring of the block $B_{2i-1}$ be copied from the block $B_{2i-2}$,
	where we remove the occurrence of the color $B_{2d+3}^i$, shift the colors to the right
	and set $B_{2i-1}^i$ to $d$. The coloring of $B_{2i}$ is copied from $B_{2i-1}$,
	but we set $B_{2i}^i$ to $B_{2d+3}^i$.
	Note that the number of blocks is enough to correctly replace all colors such that at the end the coloring of the block $B_{2d+3}$ is reached and no coloring conflict is created.

	After the coloring of the left neighborhood of $w_i$, we proceed similarly with the right neighborhood (with the difference that the block $B_{d+2}$ may now inherit the coloring from the first (greedy) phase.
	The total running time is clearly $pd^{O(d)} + O(n\log c)$.
	The additive error introduced by this coloring routine is at most the size of all neighborhoods of the precolored vertices $w_1,\dots,w_p$, which is $O(p d^2)$.
\end{proof}

\subsection{Exact algorithm for few colors}
We have already seen that \DPED is \Wone-hard when parameterized by both the number of colors and the distance, which rules out the existence of an FPT algorithm parameterized by these parameters alone under standard assumptions.
We have also seen that an FPT approximation algorithm with additive error is possible if there are few precolored vertices and the distance parameter is small, even for an unbounded number of colors.
Here, we build on these results and show that an exact FPT algorithm is possible when the number of colors, the distance parameter, and the number of precolored vertices are all small.
In some ways, it can also be seen as an extension of the greedy algorithm from Theorem~\ref{thm:greedy-algo} in the case of few colors and small distance.
Our approach exploits an intriguing connection between $d$-distance colorings of paths and regular languages.

A \emph{non-deterministic finite automaton (NFA)} $A$ is a tuple $(\Sigma, Q, \delta, q_0, F)$, where $\Sigma$ is a finite alphabet, $Q$ is a finite set of states, $q_0 \in Q$ is an initial state, $F \subseteq Q$ is the set of final states, and $\delta \subseteq Q \times \Sigma \times Q$ is a transition relation.
A word $w = w_1 \cdots w_n$ over the alphabet~$\Sigma$ is \emph{accepted by $A$} if there exists a sequence of states $p_0, \dots, p_n$ such that $p_0$ is the initial state~$q_0$, $p_n$ is some final state in $F$ and we have $(p_{i-1}, w_i, p_i) \in \delta$ for every $i \in [n]$.
We denote by $L(A)$ the language of all words accepted by~$A$.

For a word $w \in \Sigma^*$ and a letter $\alpha \in \Sigma$, we let $|w|_\alpha$ denote the number of occurrences of~$\alpha$ in~$w$.
We assume that $\Sigma = \{\alpha_1, \dots, \alpha_k\}$ and we define the \emph{Parikh image $\mathcal{P}(w)$} of a word $w \in \Sigma^*$ to be the tuple $(|w|_{\alpha_1}, \dots, |w|_{\alpha_k}) \in \mathbb{N}^k$.
In other words, the Parikh image of a word $w$ captures exactly the number of occurrences of each letter from~$\Sigma$ in~$w$.
Moreover, for a set of words $L \subseteq \Sigma^*$, we let $\mathcal{P}(L)$ denote the set of all Parikh images of words in $L$.

It is known that given a small NFA~$A$ over a small alphabet~$\Sigma$ together with a tuple~$\mathbf{b}$, it is possible to efficiently decide whether there exists a word in $w \in L(A)$ such that $\mathcal{P}(w) = \mathbf{b}$.

\begin{theorem}[\cite{KopczynskiT10, To10}]
	\label{thm:parikh-images}
	Given an NFA $\mathcal{A}$ with $n$ states over the alphabet $\Sigma$ of size~$k$, and a tuple $\mathbf{b} \in \mathbb{N}^k$, checking whether $\mathbf{b} \in \mathcal{P}(L(\mathcal{A}))$ can be done in time $2^{O(k^2 \log(kn) + \log\log b)}$ where $b = \|\mathbf{b}\|_{\infty}$.  
\end{theorem}

The main theorem of this section follows by reducing \DPED to membership for Parikh images of a regular language given by a small NFA.
In order to make the argument easier to follow, we use an intermediate problem that asks whether there exists a word in the regular language with a given Parikh image that, moreover, has predetermined letters in some positions.
Let $P$ be a relation $P \subseteq \mathbb{N} \times \Sigma$ where $\Sigma$ is a finite alphabet.
We say that a word $w$ is \emph{$P$-constrained} if we have $w_i = \alpha$ for every $(i, \alpha) \in P$.

\defproblem{\textsc{Constrained Membership for Parikh Languages of NFAs} (CMPL)}
{An NFA $A$ over alphabet $\Sigma$ of size~$k$, a tuple $\mathbf{b} \in \mathbb{N}^k$ and a relation $P \subseteq \mathbb{N} \times \Sigma$.}{Is there a $P$-constrained word $w \in L(A)$ such that $\mathcal{P}(w) = \mathbf{b}$?}

\begin{theorem}
	\label{thm:cmpl-fpt}
	\textsc{CMPL} can be solved in time  $2^{O((k+p)^2 \log((k + p) n) + \log\log b)}$ where $n$ is the number states of the NFA~$A$, $b = \|\mathbf{b}\|_{\infty}$ and $p = |P|$. 
\end{theorem}

Due to the space constraints, we include the proof of Theorem~\ref{thm:cmpl-fpt} only in Appendix~\ref{apx:cmpl-fpt}.

\begin{theorem}
	\DPED can be solved on path instances in time $O(n) + 2^{O(d \cdot (c + q)^2 \log(c + q))}\cdot (\log n)^{O(1)}$ where $c$ is the number of colors, $d$ is the distance parameter and $q$ is the number of precolored vertices. 
\end{theorem}
\begin{proof}
	Let $(G, C, d, \gamma', \eta)$ be an instance of \DPED.
	As we have observed before, we can assume that $G$ is in fact a single path on vertices $v_1, \dots, v_n$ in this order along the path.
	
	We start by showing that proper $d$-distance colorings of paths form a regular language over~$C$.
	But first, we need some notation.
	For an alphabet~$\Sigma$ and a positive integer~$t$, we let $\Sigma^{\preceq t}$ denote the set of all words of length at most~$t$ over~$\Sigma$ with no repeated letters, including the empty word~$\varepsilon$.
	Moreover for word $w \in \Sigma^{\preceq t}$ and letter $\alpha \notin w$, we let $\psi_t(w, \alpha)$ be an operation which appends the letter $\alpha$ at the end of~$w$ and then possibly deletes the first letter if $w$ already has length~$t$.
	
	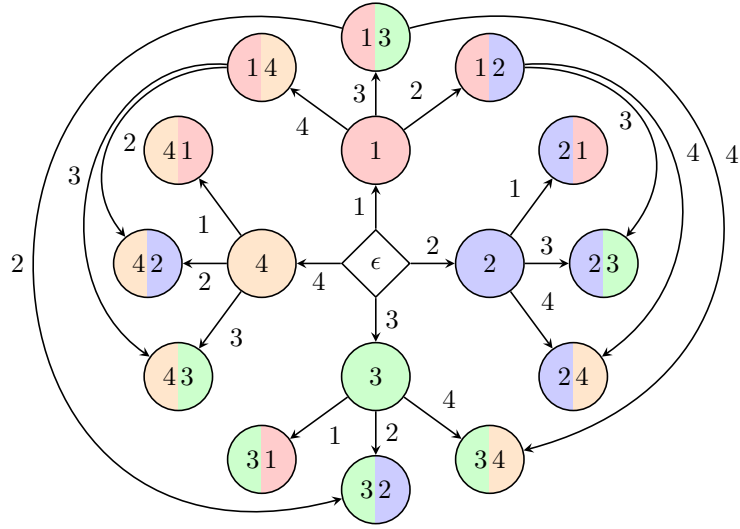
\begin{figure}
		\begin{center}
			\begin{tikzpicture}[->, >=stealth, auto, node distance=2.8cm, semithick]
				
				\tikzstyle{every state}=[fill=white,draw=black,text=black,minimum size=0.9cm]
				\tikzset{
					state1/.style={fill=red!20},
					state2/.style={fill=blue!20},
					state3/.style={fill=green!20},
					state4/.style={fill=orange!20},
					semifill/.style 2 args={path picture={
							\fill[#1] 
							(path picture bounding box.south) rectangle
							(path picture bounding box.north west);
							\fill[#2] 
							(path picture bounding box.south) rectangle
							(path picture bounding box.north east);
					}}
				}
				
				\node[state,initial by diamond] (E) at (0,0) {$\epsilon$};
				
				\def\rA{1.5}
				\def\rB{3}
				
				\node[state, state1] (1) at (90:\rA) {$1$};
				\node[state, state2] (2) at (0:\rA) {$2$};
				\node[state, state3] (3) at (270:\rA) {$3$};
				\node[state, state4] (4) at (180:\rA) {$4$};
				
				\node[state, semifill={red!20}{blue!20}] (12) at (60:\rB) {$1\,2$};
				\node[state, semifill={red!20}{green!20}] (13) at (90:\rB) {$1\,3$};
				\node[state, semifill={red!20}{orange!20}] (14) at (120:\rB) {$1\,4$};
				
				\node[state, semifill={blue!20}{red!20}] (21) at (30:\rB) {$2\,1$};
				\node[state, semifill={blue!20}{green!20}] (23) at (0:\rB) {$2\,3$};
				\node[state, semifill={blue!20}{orange!20}] (24) at (330:\rB) {$2\,4$};
				
				\node[state, semifill={green!20}{red!20}] (31) at (240:\rB) {$3\,1$};
				\node[state, semifill={green!20}{blue!20}] (32) at (270:\rB) {$3\,2$};
				\node[state, semifill={green!20}{orange!20}] (34) at (300:\rB) {$3\,4$};
				
				\node[state, semifill={orange!20}{red!20}] (41) at (150:\rB) {$4\,1$};
				\node[state, semifill={orange!20}{blue!20}] (42) at (180:\rB) {$4\,2$};
				\node[state, semifill={orange!20}{green!20}] (43) at (210:\rB) {$4\,3$};
				
				\path (E) edge node {1} (1)
				edge node {2} (2)
				edge node {3} (3)
				edge node {4} (4);
				
				\path (1) edge node {2} (12)
				edge node {3} (13)
				edge node {4} (14);
				
				\path (2) edge node {1} (21)
				edge node {3} (23)
				edge node {4} (24);
				
				\path (3) edge node {1} (31)
				edge node {2} (32)
				edge node {4} (34);
				
				\path (4) edge node {1} (41)
				edge node {2} (42)
				edge node {3} (43);
				
				\path (12) edge[bend right, out=60, in=110, looseness=1.4] node[left] {3} (23);
				\path (12) edge[bend right, out=80, in=110, looseness=1.4] node {4} (24);
				
				\path (13) edge[bend right, out=255, in=285, looseness=2.3] node[left] {2} (32);
				\path (13) edge[bend right, out=90, in=90, looseness=2] node {4} (34);
				
				\path (14) edge[bend right, out=300, in=250, looseness=1.2] node[xshift=-2] {2} (42);
				\path (14) edge[bend right, out=280, in=250, looseness=1.2] node[left] {3} (43);
				%
				%
				%
				%
				%
				%
				%
				%
				%
				%
				\pgfresetboundingbox
				\path [use as bounding box] ($(E)-(5,3.5)$) rectangle ($(E)+(5,3.5)$);
			\end{tikzpicture}
		\end{center}
		\caption{The NFA $A_{[4], 2}$ that accepts exactly the 2-distance colorings of paths using colors~$\{1, \dots, 4\}$. The initial state~$\epsilon$ is in the center, all its states are final and for improved readability, we omit all transitions from states with two colors except for $12$, $13$ and $14$.}
		\label{fig:coloring-nfa}
	\end{figure}

	Now we can define the NFA $A_{\Sigma,t} = (\Sigma, \Sigma^{\preceq t}, \delta, \varepsilon, \Sigma^{\preceq t})$ where $(w, \alpha, w') \in \delta$ if and only if $\psi_t(w, \alpha) = w'$.
	The language $L(A_{\Sigma,t})$ clearly contains exactly those words over~$\Sigma$ that do not contain two occurrences of the same letter at a distance less than~$t+1$ from each other.
	See Figure~\ref{fig:coloring-nfa}.
	Therefore, the NFA $A_{C, d}$ accepts exactly valid $d$-distance colorings of paths, and we can solve \DPED with no precolored vertices by invoking Theorem~\ref{thm:parikh-images} for the NFA $A_{C,d}$ where the tuple $\mathbf{b}$ encodes the demand function~$\eta$.
	
	Moreover, we can encode the precolored vertices as additional constraints.
	In particular, we define a set of constraints $P$ where $(i, \alpha) \in P$ if and only if the vertex~$v_i$ is precolored and $\gamma'(v_i) = \alpha$.
	Additionally, we set $\mathbf{b} \in \mathbb{N}^c$ where $b_j = \eta(j) + |\gamma'^{-1}(j)|$ for each $j \in [c]$.
	Observe that the mapping $w \mapsto \gamma_w$ where $\gamma_w(v_i) = w_i$ is a bijection between $P$-constrained words with $\mathcal{P}(w) = \mathbf{b}$ and proper distance $d$-colorings of~$G$ that extend~$\gamma'$ and satisfy the demands~$\eta$.
	In other words, $(G, C, d, \gamma', \eta)$ is yes-instance of \DPED if and only if $(A_{C,d}, \mathbf{b}, P)$ is a yes-instance of \textsc{CMPL}.

	To finish the proof, it remains to invoke the algorithm of Theorem~\ref{thm:cmpl-fpt} on the produced NFA~$A$, the tuple~$\mathbf{b}$, and the set of constraints~$P$.
	The desired runtime follows since the automaton~$A_{C,d}$ has $O(c^d)$ states and the alphabet is exactly the set of colors~$C$ of size~$c$.
\end{proof}

%% file: no-demands.tex
In this section, we focus on variants of the two problems (\textsc{Precoloring Extension} and \textsc{List Coloring}) on paths where we no longer specify demands for each color, but we still require vertices at distance at most~$d$ to receive different colors.
We refer to these problems as \textsc{Distance Precoloring Extension} (\DPE) and \textsc{Distance List Coloring} (\DLC).
Due to the space constraints, we postpone the proofs to Appendix~\ref{apx:nodemands}.

We show that \DPE is \NP-complete on paths by reduction from \textsc{Precoloring Extension} on unit interval graphs~\cite{Marx06}.
We remark that unit interval graphs are a suitable source of hardness since it is known that unit interval graph are precisely the induced subgraphs of path powers~\cite{LinRSS11}.
Consequently, \DPED remains \NP-complete on paths as well.

\begin{theorem}
\DPE is \NP-complete even when restricted to paths.
\label{thm:DPE}
\end{theorem}

\begin{corollary}
\label{cor:DPED-path}
\DPED is \NP-complete even when restricted to paths.
\end{corollary}
\begin{proof}
	Let $(G, C, d, \gamma')$ be an instance of \DPE where $G$ is a path with $n$ vertices.
	We construct an instance $(G', C, d, \gamma', \eta)$ where $G'$ is simply obtained by adding $(|C| - 1) \cdot n$ isolated vertices (paths of zero length) to~$G$ and setting $\eta(a) = n$ for every $a \in C$.
	Clearly, every $d$-distance coloring of $G'$ that extends~$\gamma'$ induced the desired $d$-distance coloring of~$G$, regardless of the demands.
	On the other hand,  we can extend any $d$-distance coloring $\gamma$ of~$G$ to~$G'$ by coloring exactly $n-|\gamma^{-1}(a)|$ of the isolated vertices with each color $a \in C$ to satisfy all demands.
	
	Finally, we construct an equivalent instance of \DPED on a single path.
	It suffices to concatenate all the disjoint paths in $G'$ to a single path while adding $d$ auxiliary vertices between each consecutive pair of original paths.
	These auxiliary vertices are then greedily precolored with $2d+1$ new auxiliary colors disjoint from~$C$ that have zero demands.
	In this way, the $d$-distance colorings of individual paths from~$G'$ are independent from each other, and the equivalence follows.
\end{proof}

In contrast, we show that the problem \textsc{Distance Precoloring Extension} without demands becomes FPT by the distance parameter~$d$ even for an unbounded number of colors.
Recall that \textsc{Distance Precoloring Extension with Demands} is \Wone-hard with respect to~$d$ by Theorem~\ref{thm:DPED-Wh}.
In fact, we show the existence of an FPT algorithm for the more general problem of \textsc{Distance List Coloring} where the task is to find a list coloring that assigns different colors to every pair of vertices at distance at most~$d$.

\begin{theorem}
\label{thm:DLC}
\textsc{Distance List Coloring} can be solved on paths in time $d^{O(d)} \cdot n$ where $n$ is the number of vertices.
\end{theorem}

%% file: appendix.tex
\section{Proofs omitted from Section~\ref{sec:wonehardness}}
\label{apx:LCD-W1h}

\begin{proof}[Proof of Proposition~\ref{pro:dp}]
	Assume that $G$ is a path with vertices~$v_1, \dots, v_n$ in this order along the path.
	The algorithm fills a table $\operatorname{DP}[i, \gamma_i, \rho]$ where $i \in [n]$, $\gamma_i \colon \{v_{\max(i-d+1,1)}, \dots, v_i\} \to C$ is an injective mapping, and $\rho: C \to [n]$ is a function.
	The semantics of the table is that $\operatorname{DP}[i, \gamma_i, \rho] = \true$ if and only if there exists a $d$-distance coloring $\gamma''$ of $\{v_1, \dots, v_i\}$ that extends both $\gamma'$ and $\gamma_i$ such that $|\gamma''^{-1}(a)| = \rho(a)$.
	Assuming these semantics, there is a solution to the input instance if and only if there exists some $\gamma_n$ such that $\operatorname{DP}[n, \gamma_n, \rho^\ast] = \true$ where $\rho^\ast(a) = \eta(a) + |\gamma'^{-1}(a)|$.
	
	Initially, we set $\operatorname{DP}[0, \emptyset, \rho] = \true$ if and only if $\rho(a) = 0$ for every color $a \in C$.
	Now, fix $i \in [n]$, $\gamma_i$ and $\rho$.
	If $v_i$ is precolored by~$\gamma'$ and $\gamma_i(v_i) \neq \gamma'(v_i)$, we immediately set $\operatorname{DP}[i, \gamma_i, \rho] = \false$ since $\gamma_i$ does not extend~$\gamma'$.
	Otherwise, let $\rho'$ be a function such that $\rho'(\gamma_i(v_i)) = \rho(\gamma_i(v_i)) - 1$ and $\rho'(a) = \rho(a)$ for every other color $a \neq \gamma_i(v_i)$.
	We set $\operatorname{DP}[i, \gamma_i, \rho] = \true$ if and only if there exists $\gamma_{i-1}$ such that
	\begin{enumerate}[(a)]
		\item $\operatorname{DP}[i-1, \gamma_{i-1}, \rho'] = \true$,
		\item the color $\gamma_{i}(v_i)$ does not appear in~$\gamma_{i-1}$, and
		\item $\gamma_i$ agrees with~$\gamma_{i-1}$ on the shared vertices.
	\end{enumerate}
	The correctness of this dynamic programming scheme follows by a standard induction.
	First, we show that if $\operatorname{DP}[i,\gamma_i,\rho] = \true$ then there exists a $d$-distance coloring $\gamma''$ of $\{v_1, \dots, v_i\}$ that extends both $\gamma'$ and $\gamma_i$ such that $|\gamma''^{-1}(a)| = \rho(a)$.
	This clearly holds when $i = 0$.
	If $\operatorname{DP}[i,\gamma_i,\rho] = \true$, then take $\gamma_i$ such that (a)--(c) are satisfied.
	By induction, there exists a suitable $d$-distance coloring $\gamma''$ of $\{v_1, \dots, v_{i-1}\}$ that can be extended to a $d$-distance coloring $\gamma^\ast$ of $\{v_1, \dots, v_{i}\}$ by setting $\gamma^\ast(v_i)=\gamma_i(v_i)$.
	It is straightforward to verify that $\gamma^\ast$ has all the desired properties.
	
	On the other hand, assume that there is a $d$-distance coloring $\gamma''$ of $\{v_1, \dots, v_i\}$ that extends both $\gamma'$ and $\gamma_i$ such that $|\gamma''^{-1}(a)| = \rho(a)$.
	It suffices to notice that $\gamma''$ restricted to $\{v_1, \dots, v_{i-1}\}$ is also $d$-distance coloring extending $\gamma'$.
	Thus by induction, we have $\operatorname{DP}[i-1, \gamma_{i-1}, \rho'] = \true$ when we take $\gamma_{i-1}$ as the restriction of~$\gamma''$.
	It follows that $\operatorname{DP}[i,\gamma_i,\rho] = \true$ since $\gamma_{i-1}$ necessarily satisfies all (a)--(c).
	
	It remains to discuss the running time.
	The table $\operatorname{DP}[\cdot,\cdot,\cdot]$ has $O(c^d \cdot n^{c+1})$ entries in total and the computation of a single entry $\operatorname{DP}[i,\gamma_i,\rho]$ takes $O(c^d)$ time by trying all possible choices of $\gamma_{i-1}$.
	We can additionally observe that $\gamma_{i-1}$ contains only one extra vertex not colored by~$\gamma_i$ and thus, the computation of $\operatorname{DP}[i,\gamma_i,\rho]$ can be done in $O(c)$ time by iterating only over the choice of color for this particular vertex.
\end{proof}

\subsection{Proof of Theorem~\ref{thm:LCD-W1h}}

We will be reducing from a multidimensional version of the subset sum problem.

\defproblem{\textsc{Multidimensional Subset Sum} (\MSS)}
{An integer~$k$, a set $S=\{s_1, \dots, s_n\}$ of item vectors with $s_i \in \mathbb{N}^k$ for every $i \in [n]$ and a target vector $t \in \mathbb{N}^k$.}
{Is there a subset $S' \subseteq S$ such that $\sum_{s \in S'} s = t$?}

The problem \MSS is known to be \Wone-hard when parameterized by $k$, even if the vectors in the input are given in unary~\cite{GanianOR21}.

\begin{proof}[Proof of Theorem~\ref{thm:LCD-W1h}]
	Let $(k, S, t)$ be an instance of \MSS.
	We construct an instance $(G, C,\allowbreak d, \gamma, \eta)$ of \LCD where the underlying graph is a disjoint union of $n$ paths, one per each vector in~$S$.
	The set of colors~$C$ is composed of a single \emph{universal color}~$a$, two \emph{fill colors} $b, b'$, $k$ \emph{target colors} $c_1, \dots, c_k$ and one \emph{auxiliary color} $\star$.
	Observe that we have $|C| = k + 4$ in total.
	
	Now we describe the \emph{item gadget}~$G_r$ for a fixed vector $r \in S$.
	Denote by $r_1, \dots, r_k$ the coordinates of $r$, i.e., we have $r = (r_1, \dots, r_k)$ and moreover, let $N_r$ denote the sum $\sum_{i=1}^k r_i$ of the coordinates.
	We also let $j_1, \dots, j_{N_r}$ be a sequence of values from $[k]$ that contains each $j \in [k]$ exactly $r_j$ times.
	We can easily construct such a sequence by concatenating $k$ blocks, each consisting of $r_j$ copies of the value~$j$ for $j$ going from~$1$ to~$k$.
	
	The graph~$G_r$ is a path on~$6 N_r$ vertices $v^r_1, \dots, v^r_{6 N_r}$ enumerated in the order along the path.
	For every $k \in [N_r]$, we set
	\begin{align*}
		L(v^r_{6k-5}) &= \{a, c_{j_k}\}\\
		L(v^r_{6k-4}) &= L(v^r_{6k-1}) = \{a, b\}\\
		L(v^r_{6k-3}) &= L(v^r_{6k}) = \{a, b'\}\\
		L(v^r_{6k-2}) &= \{a, \star\}
	\end{align*}
	Notice that the target colors $c_1, \dots, c_k$ appear only in the lists of odd-numbered vertices, while the auxiliary color appears only in the lists of even-numbered vertices.
	Moreover, for any fixed $j \in [k]$, the target color~$c_j$ appears in the lists exactly as many times as in the sequence $j_1, \dots, j_{N_r}$, which is equal to $r_j$ by definition.
	
	\begin{claim}
		\label{claim:vector-gadget}
		Let~$\gamma$ be a proper list-coloring of the gadget $G_r$ that uses the universal color $a$ exactly $3 N_r$ times.
		Then  $\gamma$ uses the colors $b$ and $b'$ both exactly $N_r$ times and moreover, we have either
		\begin{itemize}
			\item $\gamma^{-1}(a) = \{v^r_{2k} \mid k\in[N_r]\}$, the target color $c_j$ is used by~$\gamma$ exactly $r_j$ times and the color~$\star$ is not used by~$\gamma$, or
			\item $\gamma^{-1}(a) = \{v^i_{2k-1} \mid k\in[N_r]\}$, $\gamma$ uses none of the target colors and it uses the color~$\star$ exactly~$N_r$ times.
		\end{itemize}
	\end{claim}
	\begin{claimproof}
		Observe that any proper list-coloring of $G_r$ can use the universal color~$a$ at most $3 N_r$ times, and moreover, any list-coloring~$\gamma'$ with exactly $3N_r$ occurrences of~$a$ must use $a$ either on all odd-numbered or on all even-numbered vertices.
		
		First, assume that $\gamma$ colors all even-numbered vertices with color~$a$.
		On the odd-numbered vertices, we have  $\gamma(v^r_{6k-5}) = c_{j_k}$,  $\gamma(v^r_{6k-3}) = b'$ and $\gamma(v^r_{6k-1}) = b$ for each $k \in [N_r]$.
		This is exactly the first situation from the statement since for each $j \in [k]$, the value~$j$ appears in the sequence $j_1, \dots, j_{N_r}$ exactly $r_j$ times by construction.
		
		Otherwise, $\gamma$ uses the color~$a$ on all odd-numbered vertices.
		In that case, we have $\gamma(v^r_{6k-4}) = b$,  $\gamma(v^r_{6k-2}) = \star$ and $\gamma(v^r_{6k}) = b'$ for each $k \in [N_r]$, and the second situation occurs in this case. 
	\end{claimproof}
	
	The final graph is simply obtained as the union of all the item gadgets, i.e., we take $G = \bigcup_{r \in S} G_r$.
	Importantly, observe that the size of $G$ is polynomial in the size of the input instance since all numbers were given in unary.
	Finally, define the demand function $\eta$ as follows
	\begin{itemize}
		\item $\eta(a) = 3 \cdot \sum_{r \in S} |r|$, $\eta(b) = \eta(b') = \sum_{r \in S} |r|$,
		\item $\eta(c_j) = t_j$ for every $j \in [k]$, and
		\item $\eta(\star) = \sum_{r \in S} |r| - |t|$.
	\end{itemize}
	where $t_j$ denotes the $j$-th coordinate of the target vector~$t$ and $|x|$ for $x \in \mathbb{N}^k$ denotes the sum of all coordinates of the vector~$x$.
	
	Observe that the sum of all demands indeed equals the total number of vertices in~$G$, i.e., we have $\sum_{c \in C}\eta(c) = 6 \cdot \sum_{r \in S} |r|$.

	\subparagraph{Correctness (``only if'').}
	Let $(k, S, t)$ be a yes-instance of \MSS and let $S'$ be a subset of $S$ such that $\sum_{r \in S'}r = t$.
	Let us define a proper list-coloring~$\gamma$ of~$G$ that meets the required demands.
	We specify $\gamma$ on each variable gadget~$G_r$ independently.
	To that end, fix $r \in S$.
	
	We let $\gamma$ color with the universal color~$a$ all the even-numbered vertices within~$G_r$ if $r \in S'$ and all the odd-numbered vertices otherwise.
	In both cases, the coloring of the remaining vertices inside~$G_r$ is uniquely determined because all of them have a neighbor colored with~$a$ and thus, they receive the other color from their lists of size two.

	Observe that $\gamma$ is indeed a proper list-coloring of~$G$.
	It remains to verify that it meets the required demands.
	Since $\gamma$ uses the universal color~$a$ exactly $3\cdot |r|$ times within~$G_r$, we can apply Claim~\ref{claim:vector-gadget}.
	It follows that the fill colors $b, b'$ are used exactly $|r|$ times in each item gadget $G_r$ and $\gamma$ satisfies the demands for colors $a$, $b$, and $b'$.
	
	It further follows from Claim~\ref{claim:vector-gadget} that each target color $c_j$ for $j \in [k]$ appears exactly $r_j$ times within~$G_r$ if $r \in S'$ and otherwise, it is not used at all.
	This makes the total number of occurrences of $c_j$ in $\gamma$ to be $\sum_{r \in S'} r_j$, which exactly equals $t_j = \eta(c_j)$ since $S'$ is a solution to \MSS.
	
	Finally, we need to check that the extra color~$\star$ also satisfies the demands.
	However, that follows simply from the fact that the sum of all demands equals the number of vertices in~$G$ and that $\gamma$ meets the demands of all the remaining colors.
	
	\subparagraph{Correctness (``if'').}
	Now, let $(G, C, \eta, L)$ be a yes-instance of \LCD and let $\gamma$ be a corresponding proper list-coloring that meets the demands.
	Since $\gamma$ meets the demands~$\eta$, we know that the universal color is used $3 \cdot \sum_{r \in S} |r|$ times.
	However, we have already observed that $\gamma$ cannot use the color~$a$ more than $3|r|$ times within the item gadget~$G_i$.
	It follows that $\gamma$ uses the universal color exactly~$3|r|$ times within each~$G_r$.
	
	Therefore, Claim~\ref{claim:vector-gadget} applies, and we define a set $S'$ where $S'$ contains an item $r \in S$ if and only if $\gamma$ uses the color~$a$ on all even-numbered vertices in~$G_r$.
	Fix $j \in [k]$.
	Due to Claim~\ref{claim:vector-gadget}, the target color~$c_j$ appears exactly $r_j$ times within~$G_r$ if $r \in S'$ and zero times otherwise.
	The total number of occurrences of~$c_j$ can thus be expressed as $\sum_{r \in S'} r_j$, and moreover, it equals $\eta(c_j) = t_j$ since $\gamma$ is a solution to \LCD.
	In other words, we have obtained $\sum_{r \in S'} r_j = t_j$ for every $j \in [k]$ and the set~$S'$ verifies that $(k, S, t)$ is a yes-instance of \MSS.
\end{proof}

\section{Full proof of Theorem~\ref{thm:cmpl-fpt}}
\label{apx:cmpl-fpt}
\begin{proof}
	Let $A = (\Sigma, Q, \delta, q_0, F)$ where we assume without a loss of generality that the alphabet $\Sigma = [k]$ is the initial segment of natural numbers.
	We first slightly modify~$A$ to keep track of word lengths explicitly.
	The modified set of states~$Q'$ contains two copies of each state~$q \in Q \setminus \{q_0\}$, an \emph{in-state}~$q$ and an \emph{out-state} $\overline{q}$.
	Additionally, we add a single initial out-state~$\overline{q_0}$.
	The original transitions are now redirected to lead only from out-states to in-states, i.e, we define transitions~$\delta'$ where $(\overline{q}, \alpha, q') \in \delta'$ whenever $(q,\alpha, q')$.
	Finally, there is only one transition from each in-state via a new \emph{counter letter}~$k+1$, i.e., we add to~$\delta'$ the tuple $(q, k+1, \overline{q})$ for every non-initial state $q \in Q \setminus \{q_0\}$.
	We set the set of final states to be the out-states of all original final states, i.e., we set $F' = \{\overline{q} \mid q \in F\}$.
	Let us denote the constructed NFA as $A' = ([k+1], Q', \delta', \overline{q_0}, F')$.
	Observe that each word in the language $L(A')$ is obtained from some word in $L(A)$ by interleaving it with the letter~$k+1$ on every even index, and vice versa.
	Therefore, the occurrences of the letter~$k+1$ exactly count the word lengths as desired.
	
	Now, we can describe the construction of an NFA $A''$ that additionally encodes the constraints in~$P$.
	Let us denote by $m$ the sum $\sum_{i=1}^k b_i$, i.e., the length of the desired solution.
	First, observe that if $P$ contains two pairs $(i, \beta)$ and $(i, \beta')$ for $\beta \neq \beta'$, then there is no possible solution and the algorithm can immediately report that and halt.
	Otherwise, we assume that $P = \{ (i_1, \beta_1), \dots, (i_p, \beta_p)\}$ where the sequence $i_1, \dots, i_p$ of indices is strictly increasing.
	We set $m_j = i_{j} - i_{j-1}$ for each $j \in [p+1]$ where we additionally take $i_0 = 1$ and $i_{p+1} = m+1$.
	Observe that we have $m_j \ge 1$ for every $j \in \{2, \dots, p+1\}$.
	
	The final automaton will consist of $p+1$ copies of $A'$, where each will use a different counter letter.
	For $j \in [p+1]$, let $A_j$ be a copy of $A'$ where we replace the counter letter with the letter~$k+j$.
	We distinguish the states in different automata using a subscript, e.g., the automaton~$A_j$ contains states $q_j$ and $\overline{q}_j$ for each $q \in Q \setminus \{q_0\}$ and additionally the state $\overline{q_0}_j$.
	The states of the final automaton~$A''$ are simply the union of the states of $A_1, \dots, A_{p+1}$, its initial state is the initial state $\overline{q_0}_1$ in~$A_1$, and its final states are the final states in~$A_{p+1}$.
	Let us denote by $\Sigma'$ the alphabet of $A''$, i.e. $\Sigma' = [k + p + 1]$.
	
		It remains to define transitions between the automata.
	These will precisely encode the fixed letters.
	Fix $j \in [p]$.
	For every $q$ and $q'$ such that $(q, \beta_j , q') \in \delta$ is a transition in~$A$, we add the transition $(\overline{q}_j, \beta_j, q'_{j+1})$.
	In other words, we allow any transition over the letter $\beta_j$ in~$A_j$ to alternatively jump to the equivalent state in $A_{j+1}$.
	This concludes the definition of the NFA~$A''$. Observe that the number of states is $O(p \cdot n)$ and the alphabet has size~$k+p+1$.
	We connect together the languages $L(A'')$ and $L(A)$.
	
	\begin{claim}\label{claim:fpt-cmpl}
		An arbitrary word $w \in \Sigma'^*$ is accepted by~$A''$ if and only if
		\begin{enumerate}
			\item $|w| = 2t$ for some $t \in \mathbb{N}$,\label{item:fpt-cmpl-1}
			\item for each $i \in [t]$, we have $w_{2i-1} \in \Sigma$ and $w_{2i} \in \Sigma' \setminus \Sigma$\label{item:fpt-cmpl-2},
			\item the word $w_1 w_3 \cdots w_{2t-1}$ is accepted by~$A$\label{item:fpt-cmpl-3},
			\item $w$ contains each letter $k+2, \dots, k+p+1$ at least once,\label{item:fpt-cmpl-4}
			\item for every $i \in [t-1]$, we have $w_{2i} \le w_{2(i+1)} \le w_{2i}+1$ and if $w_{2i} < w_{2(i+1)}$ then $w_{2i+1} = \beta_{w_{2i}-k}$.\label{item:fpt-cmpl-5}
		\end{enumerate}
	\end{claim}
	\begin{claimproof}
		First, suppose that $w \in L(A'')$ and let $r_0, \dots, r_{|w|+1}$ be a sequence of states of~$A''$ witnessing that $w$ is accepted by~$A'$.
		It follows immediately from the properties of~$A''$ that $r_i$ is an out-state if and only if $i$ is even, which verifies property~\ref{item:fpt-cmpl-2}.
		The final states of $A''$ are only out-states and thus, $w$ must be of even length $2t$ for some $t \in \mathbb{N}$ (property~\ref{item:fpt-cmpl-1}).
		Moreover, the sequence of visited out-states verifies that the word $w_1 w_3 \cdots w_{2t-1}$ is accepted by the original automaton~$A$ (property~\ref{item:fpt-cmpl-3}).
		
		Since the sequence of states starts in $A_1$ and ends in $A_{p+1}$ it must consecutively traverse through all the automata $A_1, \dots, A_{p+1}$ and thus, we have $w_{2i} \le w_{2(i+1)}$ for every $i\in[t-1]$.
		Moreover, at least one transition between some in-state $q_i$ and out-state $\overline{q}_i$ must have occurred in each automaton~$A_i$ for $i \ge 2$, which verifies property~\ref{item:fpt-cmpl-4} and the first half of property~\ref{item:fpt-cmpl-5}.
		Note that this does not necessarily hold for~$A_1$ since the first transition might have been from the initial state~$\overline{q_0}_1$ directly to some in-state $q'_2$ in the automaton~$A_2$.
		Whenever a transition jumps from~$A_i$ to~$A_{i+1}$, the corresponding letter in~$w$ must match the~$i$th constraint, i.e., it must be equal to~$\beta_i$.
		This implies the second half of property~\ref{item:fpt-cmpl-5}, which wraps up the proof of this first implication.
		
		For the other direction, assume a word $w$ of length~$2t$ satisfies all properties~\ref{item:fpt-cmpl-1}--\ref{item:fpt-cmpl-5}.
		Let $r_0, \dots, r_t$ be a sequence of states verifying that~$A$ accepts the word $w_1 w_3 \cdots w_{2t-1}$.
		We define a sequence of states $r'_0, \dots, r'_{2t}$ in~$A''$ and our goal is to show that it is an accepting sequence for the word~$w$.
		First, we set $r'_0$ to be the initial state $\overline{r_0}_1$ within~$A_1$.
		For $i \in [t]$, we set $r'_{2i-1} = (r_{i})_j$ and $r'_{2i} = \overline{r_{i}}_j$ where $j = w_{2i}-k$.
		
		Clearly, the tuple $(r'_{2i-1}, w_{2i}, r'_{2i})$ is a valid transition for every $i \in [t]$ because it goes between an in-state and the corresponding out-state within the automaton~$A_j$ where $j=w_{2i}-k$.
		We distinguish two possibilities for the tuple $(r'_{2i-2}, w_{2i-1}, r'_{2i-1})$ for each $i \in [t]$.
		Either $r'_{2i-2}$ and $r'_{2i-1}$ are states within the same automaton~$A_j$ for some~$j \in [p+1]$, in which case $(r'_{2i-2}, w_{2i-1}, r'_{2i-1})$ is a valid transition in~$A''$ since $(r_{i-1}, w_{2i-1}, r'_{i})$ was a valid transition in the original automaton~$A$.
		Otherwise, $r'_{2i-2}$ and $r'_{2i-1}$ lie in automata~$A_j$ and~$A_{j+1}$ for some $j \in [p]$, respectively, by property~\ref{item:fpt-cmpl-4} and moreover, $w_{2i-1}$ is exactly the letter~$\beta_j$ by property~\ref{item:fpt-cmpl-5}.
		Observe that this exactly agrees with how we defined the transitions from $A_j$ to $A_{j+1}$, and $(r'_{2i-2}, w_{2i-1}, r'_{2i-1})$ is a valid transition in~$A''$ for every~$i \in [t]$.
	\end{claimproof}
	
	In order to reduce the problem to the membership for Parikh images, we also need to define the target vector~$\mathbf{b}'$.
	We set $b'_i = b_i$ for every $i \in [k]$ and $b'_{k + j} = m_j$ for every $j \in [p+1]$.
	Observe that we have $\sum_{i=1}^{c + p + 1}b'_i = 2m$.
	
	\begin{claim}
		\label{claim:cmpl-to-parikh}
		There is a $P$-constrained word $w \in L(A)$ such that $\mathcal{P}(w) = \mathbf{b}$ if and only if there is a word $w' \in L(A')$ such that $\mathcal{P}(w') = \mathbf{b}'$.
	\end{claim}
	
	\begin{claimproof}
		First, suppose that there is a $P$-constrained word $w \in L(A)$ of length~$m$ such that $\mathcal{P}(w) = \mathbf{b}$.
		We define a word $w'$ of length~$2m$ by inserting letters from the set $\{k+1, \dots, k+p+1\}$ into $w'$ on every other position in such a way that (i) there are exactly $m_j$ occurrences of the letter $k + j$ for each $j \in [p+1]$, and (ii) the occurrences of these extra letters occur in increasing order of their values.
		Formally for each $i \in [m]$, we set $w_{2i-1} = w'_i$ and $w_{2i} = k + t$ where $t$ is the smallest integer such that $\sum_{j = 1}^t m_j \ge i$, or equivalently, $i_t - i_0 = i_t - 1 \ge i$.
		In this way, we guarantee that $\mathcal{P}(w') = \mathbf{b'}$.
		Moreover, we have $w' \in L(A'')$ by Claim~\ref{claim:fpt-cmpl} since $w'$ satisfies properties~\ref{item:fpt-cmpl-1}--\ref{item:fpt-cmpl-4} therein by definition and property~\ref{item:fpt-cmpl-5} follows from the fact that $w$ is $P$-constrained.
		
		Now we assume that there is a word $w' \in L(A'')$ such that $\mathcal{P}(w') = \mathbf{b'}$.
		The word $w'$ must satisfies all properties \ref{item:fpt-cmpl-1}--\ref{item:fpt-cmpl-5} from Claim~\ref{claim:fpt-cmpl}.
		Let $w$ be the word obtained by restricting $w'$ to odd-numbered letters.
		By property~\ref{item:fpt-cmpl-3}, we know that $w \in L(A)$ and by property~\ref{item:fpt-cmpl-2}, we get that $\mathcal{P}(w) = \mathbf{b}$ since $w$ equals $w'$ restricted to the letters from~$[k]$.
		It remains to show that $w$ is $P$-constrained.
		However, that follows from property~\ref{item:fpt-cmpl-5} due to the way we have set up~$\mathbf{b'}$ for the letters $k+1, \dots, k+p+1$.
	\end{claimproof}

	The theorem follows by invoking the algorithm of Theorem~\ref{thm:parikh-images} on the NFA~$A''$ and the tuple~$\mathbf{b}'$.
\end{proof}

\section{Proofs omitted from Section~\ref{sec:nodemands}}
\label{apx:nodemands}
\subsection{Proof of Theorem~\ref{thm:DPE}}
\begin{proof}
We reduce from \textsc{Precoloring Extension} that is known to be \NP-hard even when restricted to unit interval graphs~\cite{Marx06}.
Let $(G, C, \gamma')$ be an instance of \textsc{Precoloring Extension}, where $G$ is a unit interval graph on $n$ vertices and $\gamma': A \to  C$ for some $A \subseteq V(G)$ is a precoloring.
Any unit interval graph admits a representation where all the endpoints of intervals are pairwise different integers from the set $\{0, \dots, n^2\}$, each interval has length exactly~$n$, and moreover, such a representation can be computed from~$G$ in linear time~\cite{CorneilKNOS95}.

We start by computing such a representation, that is, each vertex $v \in V(G)$ has an associated interval $[v_L, v_R]$ where both $v_L, v_R \in \{0, \dots, n^2\}$ and $v_R - v_L = n$.
Furthermore, the intervals of two different vertices have completely disjoint endpoints.

Now, we describe the construction of an instance $(G', C', \beta', d)$ of \DPE where $G$ is a path.
First, we set $d = 3n$, $\ell = 3n^2 + 2d$ and we take $G'$ to be a path of length $\ell$ on vertices $w_0, w_1, \dots, w_{\ell}$ in this order along the path.
We also define a map $\varphi\colon \{0, \dots, n^2\} \to V(G')$ where $\varphi(i) = w_{3i + d}$ and we refer to the vertices in the image $\varphi(\{0, \dots, n^2\})$ as the \emph{main vertices}.
First, we notice that the main vertices are somewhat sparsely distributed.

\begin{claim}
\label{claim:non-main-dense}
For an arbitrary main vertex $w_i$, there exist at least $d+1$ pairwise different non-main vertices at distance at most~$d$ from~$w_i$.
\end{claim}
\begin{proof}
Any consecutive segment of~$d$ vertices contains at most $\frac{d}{3}$ main vertices since a vertex~$w_j$ can lie in the image of~$\varphi$ only when $j$ is divisible by three.
Moreover, we ensured that all main vertices lie at a distance of at least $d$ from both endpoints of the path.
Thus for a main vertex~$w_i$, there are at least~$\frac{2d}{3}$ non-main vertices in both of the segments $w_{i+1}, \dots, w_{i+d}$ and $w_{i-d}, \dots, w_{i-1}$.
It remains to observe that the inequality $\frac{4d}{3} \ge d+1$ holds whenever $d \ge 3$, which is equivalent to $n \ge 1$.
\end{proof}

Furthermore, we define a subset of the main vertices $S = \{\varphi(v_L) \mid v \in V(G)\}$ which will be referred to as \emph{representative vertices}.
Recall that a \emph{$d$-th power} of a graph $H$ is a graph $H^d$ on the same vertex set $V(H)$ that contains an edge between two vertices $u$ and~$v$ if and only if $\dist_H(u,v) \le d$.

\begin{claim}
\label{claim:representative-induced}
	The set $S$ induces a copy of $G$ in the $d$-th power of~$G'$.
\end{claim}
\begin{claimproof}
	We claim that the map $\rho\colon G \to (G')^d[S]$ defined as $\rho(v) = \varphi(v_L)$ is an isomorphism.
	First, observe that $\rho$ is injective since all the endpoints of intervals are pairwise different and $\varphi$ itself is injective.
	It remains to notice that there is an edge between two vertices $u$ and $v$ in~$G$ if and only if $|u_L - v_L| \le n$ which is if and only if $\dist_{G'}(\varphi(u), \varphi(v)) \le 3n = d$.
\end{claimproof}

We define the set of colors $C'$ as $C \cup C^\star$ where $C^\star = \{c^\star_0, \dots, c^\star_d\}$ is a set of $d+1$ \emph{auxiliary colors} completely disjoint from~$C$.
Let $j_1< j_2 \dots < j_t$ be the ordered sequence of indices of all non-representative vertices in~$G'$ where $t = |V(G')| - |S|$.
We let $\beta^\star\colon V(G')\setminus S \to C^\star$ be a coloring defined as $\gamma^\star(w_{j_k}) = c^\star_{k \bmod (d+1)}$ for every $k \in [t]$.
In other words, $\beta^\star$ cyclically iterates through the set~$C^\star$ of auxiliary colors but only on the non-representative vertices.
Observe that $\beta^\star$ is a proper partial $d$-coloring because a monochromatic pair of vertices $w_i$, $w_j$ can occur at distance at least $d+1$ and that happens if only if there is no representative vertex in the segment $w_{i+1}, \dots, w_{j-1}$.

Let $B \subseteq S$ be the set of all representative vertices that correspond to the precolored vertices of~$G$, i.e., $B = \{\varphi(v_L) \mid v \in A\}$.
The set of precolored vertices in $G'$ is going to be $A' = (V(G') \setminus S) \cup B$ and we define the precoloring $\beta'\colon A' \to C'$ as follows
\[
\beta'(w_i) = \begin{cases}
	\beta^\star(w_i) &\text{if $w_i \notin S$,}\\
	\gamma'(v) &\text{if $w_i \in B$ and $w_i = \varphi(v_L)$ for $v \in A$.}
\end{cases}
\]

This concludes the construction of the instance.

\subparagraph{Correctness (``only if'').}
Let $(G, C, \gamma')$ be a yes-instance of \textsc{Precoloring Extension} and let $\gamma \colon V(G) \to C$ be a proper coloring that extends~$\gamma'$.
We produce a coloring of~$G'$ that mirrors $\gamma$ on the representative vertices and elsewhere, it equals the precoloring~$\beta'$.
Formally, let $\beta\colon V(G') \to C'$ where
\[
\beta(w_i) = \begin{cases}
	\gamma(v) &\text{if $w_i \in S$ and $w_i = \varphi(v_L)$,}\\
	\beta'(w_i) &\text{otherwise.}
\end{cases}
\]

Clearly, $\beta$ extends the precoloring~$\beta'$, and it suffices to verify that $\beta$ is a proper $d$-distance coloring.
Observe that $\beta$ uses disjoint sets of colors on the representative and non-representative vertices, respectively.
Restricted to non-representative vertices, $\beta$ is exactly equal to~$\beta^\star$, which does not contain any monochromatic pair of vertices at distance at most~$d$.
And if $\beta$ contained a monochromatic pair of non-representative vertices at distance at most~$d$, then $G$ would contain a monochromatic pair of adjacent vertices by Claim~\ref{claim:representative-induced}.

\subparagraph{Correctness (``if'').}
Now,  assume that $(G', C', \beta', d)$ is a yes-instance of \DPE and let $\beta\colon V(G') \to C'$ be a proper $d$-distance coloring that extends~$\beta'$.
We simply set~$\gamma\colon V(G) \to C'$ to be the coloring induced by~$\beta$ on the representative vertices, i.e., we set $\gamma(v) = \beta(\varphi(v_L))$.
Claim~\ref{claim:representative-induced} directly implies that~$\gamma$ is a proper coloring of~$G$, and thus, it remains to verify that~$\gamma$ does not use any auxiliary color.

Assume for a contradiction that there is a vertex~$v \in V(G)$ such that $\beta(\varphi(v_L)) = c^\star_j$ for some  $j \in \{0, \dots, d\}$.
By Claim~\ref{claim:non-main-dense}, there are at least~$d+1$ non-main vertices at distance at most~$d$ from~$\varphi(v_L)$.
On these non-main vertices, the precoloring $\beta'$ is by definition equal to the auxiliary coloring~$\beta^\star$ which iterates cyclically through the auxiliary colors $C^\star$.
It follows that there must be a non-main vertex~$w_i$ at distance at most $d$ from~$\varphi(v_L)$ such that $\beta(w_i) = \beta^\star(w_i) = c^\star_j$ and we reached a contradiction with $\beta$ being a proper $d$-distance coloring.
\end{proof}

%

\subsection{Proof of Theorem~\ref{thm:DLC}}
\begin{proof}
Let $G$ be the input path on~$n$ vertices $v_1, \dots, v_n$ in this order along the path.
The crucial observation that given an arbitrary $d$-distance list pre-coloring $\gamma' \colon V(G)\setminus \{v_i\}$ for any~$i \in [n]$, there are at most $2d$ colors that cannot be used at vertex~$v_i$ to extend $\gamma'$ to a proper $d$-distance list coloring of~$G$.
It follows that whenever a vertex~$v$ has more than $2d+1$ available colors in its list~$L(v)$, we can safely prune the list down to arbitrary~$2d+1$ colors without affecting the feasibility of the solution.
We remark that this is where the argument breaks down if we would also consider demands.
In the remainder of the proof, we describe a dynamic programming algorithm that finds a feasible $d$-distance list coloring $\gamma$ in time $\ell^d \cdot n$ where $\ell = \max_i |L(v_i)|$.
The result follows by running this algorithm on the input graph~$G$ with lists arbitrarily pruned down to size at most~$2d+1$.

To simplify the following, let us denote by $N_i$ the vertices $v_j, \dots, v_i$ with $j = \max(i-d+1, 1)$.
Notice that when coloring~$G$ in the order along the path, the vertex~$v_i$ can be colored using a color $c \in C$ if and only if $c$ is not already used on the set $N_{i-1}$.
The algorithm inductively computes for each $i \in [n]$, the set $S_i$ of all the possible list colorings of vertices in~$N_i$ that can be extended to a full $d$-distance list coloring of the path $v_1, \dots, v_i$.

We first set $S_0$  to contain the empty coloring.
For $i \in [n]$, we let $S_i$ contain a particular partial coloring $\gamma_i \colon G[N_i] \to C$ if and only if there exists $\gamma_{i-1} \in S_{i-1}$ such that (i) $\gamma_i$ and $\gamma_{i-1}$ agree on $N_i \cap N_{i-1}$, (ii) the color $\gamma_i(v_i)$ is not used by~$\gamma_{i-1}$ and (iii) $\gamma_i(v_i) \in L(v_i)$.
There exists a $d$-distance list coloring of the graph~$G$ if and only if $S_n$ is non-empty.
It follows by a standard induction on~$i$ that there exists a $d$-distance list coloring~$\gamma$ of the path $v_1, \dots, v_i$ if and only if there is some partial coloring $\gamma_i \in S_i$ where $\gamma$ extends $\gamma_i$.

It remains to discuss the running time.
The total number of list colorings of~$N_{i+1}$ is at most $(2d+1)^d$ because each list contains at most~$2d+1$ colors.
Moreover, there are at most $2d+1$ choices for the color of $v_i$ and for each, we verify the condition (ii) in time~$O(d)$.
It follows that the set $S_i$ can be computed from the set~$S_{i-1}$ in time $d^{O(d)}$ and the whole computation takes $d^{O(d)} \cdot n$ time as promised.
\end{proof}